\documentclass[english]{article}
\usepackage[T1]{fontenc}
\usepackage[latin9]{inputenc}
\usepackage{geometry}
\usepackage{xcolor}
\usepackage{units}
\usepackage{amsmath}
\usepackage{amsthm}
\usepackage{amssymb}
\usepackage{stmaryrd}
\PassOptionsToPackage{normalem}{ulem}
\usepackage{ulem}

\makeatletter

\providecommand{\tabularnewline}{\\}
\providecolor{lyxadded}{rgb}{0,0,1}
\providecolor{lyxdeleted}{rgb}{1,0,0}

\DeclareRobustCommand{\lyxsout}[1]{\ifx\\#1\else\sout{#1}\fi}

\newcommand{\lyxaddress}[1]{
	\par {\raggedright #1
	\vspace{1.4em}
	\noindent\par}
}
\theoremstyle{plain}
\newtheorem{thm}{\protect\theoremname}
\theoremstyle{definition}
\newtheorem{defn}[thm]{\protect\definitionname}
\theoremstyle{plain}
\newtheorem{lem}[thm]{\protect\lemmaname}
\theoremstyle{definition}
\newtheorem{example}[thm]{\protect\examplename}
\theoremstyle{plain}
\newtheorem{cor}[thm]{\protect\corollaryname}

\@ifundefined{date}{}{\date{}}
\usepackage{nicefrac}

\usepackage{tikz}

\usepackage{multirow}

\usepackage{geometry}

\usepackage{cite}

\date{}

\usepackage{babel}

\makeatother

\usepackage{babel}
\providecommand{\corollaryname}{Corollary}
\providecommand{\definitionname}{Definition}
\providecommand{\examplename}{Example}
\providecommand{\lemmaname}{Lemma}
\providecommand{\theoremname}{Theorem}

\begin{document}
\title{Contextuality and dichotomizations of random variables}
\author{Janne V.\ Kujala\textsuperscript{1} and Ehtibar N.\ Dzhafarov\textsuperscript{2}}
\maketitle

\lyxaddress{\begin{center}
\textsuperscript{1}University of Turku, Finland, jvk@iki.fi\\\textsuperscript{2}Purdue
University, USA, ehtibar@purdue.edu
\par\end{center}}
\begin{abstract}
The Contextuality-by-Default approach to determining and measuring
the (non)contextuality of a system of random variables requires that
every random variable in the system be represented by an equivalent
set of dichotomous random variables. In this paper we present general
principles that justify the use of dichotomizations and determine
their choice. The main idea in choosing dichotomizations is that if
the set of possible values of a random variable is endowed with a
pre-topology (V-space), then the allowable dichotomizations split
the space of possible values into two linked subsets (``linkedness''
being a weak form of pre-topological connectedness). We primarily
focus on two types of random variables most often encountered in practice:
categorical and real-valued ones (including continuous random variables,
greatly underrepresented in the contextuality literature). A categorical
variable (one with a finite number of unordered values) is represented
by all of its possible dichotomizations. If the values of a random
variable are real numbers, then they are dichotomized by intervals
above and below a variable cut point.
\end{abstract}

\section{Introduction}

This paper deals with \emph{systems of random variables} 
\begin{equation}
\mathcal{R}=\left\{ R_{q}^{c}:c\in C,q\in Q,q\Yleft c\right\} ,\label{eq:system}
\end{equation}
where $Q$ denotes the set of properties $q$ being measured (generically
referred to as \emph{contents} of the random variables), $C$ denotes
the set of conditions under which the measurements are made (referred
to as \emph{contexts} of the random variables), and the relation $\Yleft$,
a subset of $Q\times C$, indicates which content is measured in which
context. As an example, consider the system of random variables describing
Bohm's version of the Einstein-Podolsky-Rosen experiment (EPR/B) \cite{Bohm1957},
the one for which Bell derived his celebrated inequalities \cite{Bell1966,CHSH1969,Bell1964}:
\begin{equation}
\begin{array}{|c|c|c|c||c|}
\hline R_{1}^{1} & R_{2}^{1} &  &  & c=1\\
\hline  & R_{2}^{2} & R_{3}^{2} &  & c=2\\
\hline  &  & R_{3}^{3} & R_{4}^{3} & c=3\\
\hline R_{1}^{4} &  &  & R_{4}^{4} & c=4\\
\hline\hline q=1 & q=2 & q=3 & q=4 & \text{\ensuremath{\mathcal{R}}}
\\\hline \end{array}\,.\label{eq:1}
\end{equation}
Here, the random variables represent measurements of spins of two
entangled spin-$\nicefrac{1}{2}$ particles, one measured by Alice
along the axis $q=1$ or $3$, and the other measured by Bob along
the axis $q=2$ or $4$. The contexts $c$ here are defined by the
four combinations of the Alice-Bob choices of axes, and we have $\left(q=1\right)\Yleft\left(c=1\right)$,
$\left(q=2\right)\Yleft\left(c=1\right)$, $\left(q=2\right)\Yleft\left(c=2\right)$,
etc. More generally, contexts $c$ may be defined by any systematically
varied conditions under which measurements are made. Thus, in the
system
\begin{equation}
\begin{array}{|c|c||c|}
\hline R_{1}^{1} & R_{2}^{1} & c=1\\
\hline R_{1}^{2} & R_{2}^{2} & c=2\\
\hline\hline q=1 & q=2 & \mathcal{R}
\\\hline \end{array}\label{eq:order system}
\end{equation}
the contexts $c=1$ and $c=2$ may represent two orders in which two
measurements, of $q=1$ and of $q=2$, are performed: $1\rightarrow2$
and $2\rightarrow1$. In this case, every $q$ is measured in every
$c$.

The sets of contexts and contents, $C$ and $Q$, can be infinite
and even uncountable, although in all practical applications known
to us $C$ is finite. The systems $\mathcal{R}$ are classified into
\emph{contextual} and \emph{noncontextual} ones. The traditional approaches
to contextuality are confined to systems without disturbance, i.e.,
those in which any two $R_{q}^{c}$ and $R_{q}^{c'}$ have identical
distributions. We call such systems \emph{consistently connected}
\cite{DKL2015}. In fact, in the traditional analysis one usually
assumes that consistent connectedness holds in the \emph{strong form}:
for any pair of contexts $c$, $c\prime,$ 
\begin{equation}
\left\{ R_{q}^{c}:q\in Q,q\prec c,c'\right\} \overset{d}{=}\left\{ R_{q}^{c'}:q\in Q,q\prec c,c'\right\} ,
\end{equation}
where $\overset{d}{=}$ means ``have the same distribution.'' That
is, the joint distributions for identically subscripted random variables
are identical \cite{AbramBarbMans2011,AbramskyBrand2011}. For instance,
in system (\ref{eq:order system}), strong consistent connectedness
means that the two distributions, in $c=1$ and $c=2$, are identical:
$\left\{ R_{1}^{1},R_{2}^{1}\right\} \overset{d}{=}\left\{ R_{1}^{2},R_{2}^{2}\right\} $.

Our approach, however, called Contextuality-by-Default (CbD) \cite{Dzh2017Nothing,Dzh2019,DzhCerKuj2017,DzhKujFoundations2017,KujDzhLar2015,KujDzhMeasures,KujDzhProof2016,DKC2020},
also applies to systems with disturbance (e.g., signaling ones), generically
referred to as \emph{inconsistently connected} systems. The reason
for this is that, both in quantum physics and in non-physical applications,
inconsistently connected systems are abundant. Declaring them all
contextual or denying the applicability to them of the notion of contextuality
seems unreasonably restrictive, as this leaves important empirical
phenomena outside contextuality analysis. For instance, system (\ref{eq:1})
describes not only the EPR/B experiment, but also a single photon
two-slit experiment. In this application $q=1$ and $3$ stand for
the left slit open and closed, respectively, and $q=2$ and $4$ stand
for the right slit open and closed, respectively. The random variables
$R_{q}^{c}$ are binary, indicating whether the photon in a given
trial hits a localized detector having passed through $q$ when the
two slits are in a particular closed-open arrangement $c$. For example,
$R_{q=2}^{c=2}=1$ if the particle passes through the open right slit
and hits the detector when the left slit is closed. This system is
inconsistently connected, and its CbD analysis shows that it is noncontextual
\cite{DKTwoSlit}. By contrast, a three-slit single particle experiment,
as shown in the same paper, is described by an inconsistently connected
system that can be contextual or noncontextual depending on specific
distributions of the random variables.

To give another example, system (\ref{eq:order system}) can describe
two sequential projective measurements performed on a single particle,
and then the system can be easily shown to be inconsistently connected
(and CbD analysis shows it is noncontextual \cite{DzhZhaKuj2016}).
Outside quantum mechanics, system (\ref{eq:order system}) describes
an important behavioral phenomenon called ``question order effect''
\cite{Wang}. Mathematically, this phenomenon is precisely the inconsistent
connectedness. In this application, $q=1$ and $2$ are two questions
that can be asked of a responder in one of two possible orders.

A practical benefit offered by CbD compared to traditional approaches
to contextuality is the ability to analyze real experiments, in which
inconsistent connectedness is present either due to the nature of
the experimental object, or due to unavoidable or inadvertent design
biases. Thus, an important quantum-mechanical experiment \cite{Lapkiewicz2015}
aimed at testing the contextuality inequalities for cyclic systems
of rank 5 \cite{KCBS2008} (cyclic systems will be defined below)
exhibits two violations of consistent connectedness, one of them expected,
the other inadvertent. This makes the traditional theory of contextuality
inapplicable without elaborate work-arounds. The CbD analysis of this
experiment \cite{KujDzhLar2015} faces no such difficulty, and demonstrates
contextuality in these data with no ``corrections'' thereof involved.
Several other applications of CbD to quantum-mechanical experiments
can be found in the literature, e.g. \cite{Ariasetal.2015,Crespietal.2017,Fluhmannetal2018,Zhanetal.2017,Leupoldetal.2018,Malinowskietal.}.
Bacciagaluppi \cite{bacciagaluppi2,Bacciagaluppi2015} used CbD to
study the Leggett-Garg paradigm \cite{LeggGarg1985} where ``signaling
in time'' typically leads to inconsistent connectedness \cite{KoflerBrukner2013,BudroniBook2016}.

In human behavior (including decision making and psychophysical judgments)
inconsistent connectedness is universal. Numerous attempts to demonstrate
contextuality in behavioral and social systems (reviewed in \cite{DzhZhaKuj2016})
have failed because they overlooked or could not properly handle this
fact. Contextuality in some systems of random variables describing
human behavior was, however, unambiguously demonstrated in recent
experiments \cite{CervDzh2018,CervDzh2019,Basievaetall2018}.

This paper focuses on a particular aspect of CbD, one that has not
been sufficiently elaborated previously. Namely, CbD requires that
in contextuality analysis of arbitrary systems every non-binary random
variable $R_{q}^{c}$ should be dichotomized, replaced with a set
of jointly distributed binary variables. We explain in this paper
why this should be done and how one is to choose the set of such dichotomizations.
We do this by systematically introducing the basics of CbD and relating
them to several principles or desiderata for an acceptable theory
of contextuality. In the process, we also explain other features of
CbD, such as the use of multimaximally connected couplings.

Let us explain the terminology. The distribution of each random variables
$R_{q}^{c}$ shows the probabilities of various measurable subsets
of the set $E_{q}$ of possible values of $R_{q}^{c}$. The types
of the sets $E_{q}$ endowed with measurable subsets are virtually
unlimited: the elements of $E_{q}$ can be numbers, functions, sets,
etc. It is, however, always possible to present $R_{q}^{c}$ by a
set of jointly distributed \emph{binary} random variables, those attaining
values 0 and 1. Indeed, for every measurable subset $A$ of $E_{q}$
one can form a random variable 
\begin{equation}
R_{q,A}^{c}=\left[R_{q}^{c}\in A\right]:=\begin{cases}
1 & \textnormal{if }R_{q}^{c}\in A\\
0 & \text{otherwise}
\end{cases}.\label{eq:dichtomization}
\end{equation}
The joint distribution of $\left\{ R_{q,A}^{c}:A\in\Sigma_{q}\right\} $,
where $\Sigma_{q}$ is the sigma-algebra on $E_{q}$, is uniquely
determined by and uniquely determines the distribution of $R_{q}^{c}$.
The binary variables $R_{q,A}^{c}$ are called \emph{dichotomizations}
of $R_{q}^{c}$, and $\left\{ A,E_{q}-A\right\} $ is called a dichotomization
of $E_{q}$. We can agree not to distinguish $R_{q,A}^{c}$ and $R_{q,E_{q}-A}^{c}$,
and also exclude $A=\emptyset$ and $A=E_{q}$, for obvious reasons.

The problem of choice arises because in most cases $\left\{ R_{q,A}^{c}:A\in\Sigma_{q}\right\} $
is too large a set of dichotomizations, and one can equivalently represent
$R_{q}^{c}$ by much smaller sets $\left\{ R_{q,A}^{c}:A\in\varUpsilon_{q}\right\} $,
with $\varUpsilon_{q}$ a proper subset of $\Sigma_{q}$. For instance,
if a random variable $R_{q}^{c}$ is absolutely continuous with respect
to the usual Lebesgue measure, the set of possible dichotomizations
includes $R_{q,A}^{c}$ for all Borel-measurable $A$ (or ``one half''
of them, as we do not distinguish $A$ and $E_{q}-A$). However, as
shown in \cite{DzhCerKuj2017}, using this set would lead to the disappointing
conclusion that all inconsistently connected systems comprising such
random variables are contextual (contravening thereby the Analyticity
principle formulated in Section \ref{sec:Contextuality}). One can
do much better by observing that the distribution of such a variable
is uniquely described by its distribution function 
\begin{equation}
F_{q}^{c}\left(x\right)=\Pr\left[R_{q}^{c}\leq x\right],
\end{equation}
whence it follows that $R_{q}^{c}$ can be equivalently represented
by a much smaller set of the variables
\begin{equation}
R_{q,(-\infty,x]}^{c}=\left[R_{q}^{c}\leq x\right].
\end{equation}
We will see that this choice of dichotomizations is dictated by the
general principles formulated in Section \ref{sec:General-principles}.
The theory of continuous and other real-valued random variables is
discussed in Section \ref{sec:cuts}.

Another class of random variables that plays an important role in
contextuality analysis is the class of \emph{categorical} variables,
those with a finite set of values that are arbitrary labels, with
no ordering. Let, e.g., $R_{q}^{c}$ have the probability mass function
\begin{equation}
\begin{array}{ccccc}
value: & 1 & 2 & 3 & 4\\
probability: & p_{1} & p_{2} & p_{3} & p_{4}
\end{array}.
\end{equation}
It has 7 distinct dichotomizations,
\begin{equation}
R_{q,\left\{ 1\right\} }^{c},R_{q,\left\{ 2\right\} }^{c},R_{q,\left\{ 3\right\} }^{c},R_{q,\left\{ 4\right\} }^{c},R_{q,\left\{ 1,2\right\} }^{c},R_{q,\left\{ 2,3\right\} }^{c},R_{q,\left\{ 1,3\right\} }^{c},\label{eq:all dichotomizations}
\end{equation}
but $R_{q}^{c}$ can also be presented by a subset of them, say,
\begin{equation}
R_{q,\left\{ 1\right\} }^{c},R_{q,\left\{ 2\right\} }^{c},R_{q,\left\{ 3\right\} }^{c},
\end{equation}
with the joint distribution
\begin{equation}
\begin{array}{cccc}
R_{q,\left\{ 1\right\} }^{c} & R_{q,\left\{ 2\right\} }^{c} & R_{q,\left\{ 3\right\} }^{c} & probabilty\\
1 & 0 & 0 & p_{1}\\
0 & 1 & 0 & p_{2}\\
0 & 0 & 1 & p_{3}\\
0 & 0 & 0 & p_{4}\\
\vdots & \vdots & \vdots & 0
\end{array}.
\end{equation}
The theory of categorical random variables is discussed in Section
\ref{sec:categorical}. The application of the general principles
here leads one to choose the complete set (\ref{eq:all dichotomizations})
of dichotomizations over its subsets.

\section{\label{sec:Probabilistic-contextuality}Systems and their couplings}

The random variables\emph{ }$R_{q}^{c}$ in (\ref{eq:system}) are
measurable functions
\begin{equation}
R_{q}^{c}:\Omega_{c}\to E_{q},
\end{equation}
which implies that for any given context $c\in C$, the random variables
in
\begin{equation}
R^{c}:=\{R_{q}^{c}:q\in Q,\,q\Yleft c\}
\end{equation}
are jointly distributed (defined on the same sample space $\Omega_{c}$).\footnote{We use script letters, $\mathcal{R},\mathcal{R}_{q},$ etc., for a
set of random variables if they are not necessarily jointly distributed.
If, as in $R^{c}$, all elements of a set are jointly distributed,
then $R^{c}$ is a random variable in its own right, and we can use
ordinary italics.} The set of random variables $R^{c}$ is called a \emph{bunch} (intuitively,
the variables measured jointly).

For a given content $q\in Q$, the set
\begin{equation}
\mathcal{R}_{q}:=\{R_{q}^{c}:c\in C,\,q\Yleft c\}
\end{equation}
is called a \emph{connection}; the random variables in a connection
take their values in the same space $E_{q}$, endowed with the same
sigma-algebra $\Sigma_{q}$. The random variables in a connection
are \emph{stochastically unrelated} (defined on different sample spaces
$\Omega_{c}$).

A powerful way to investigate relations between stochastically unrelated
variables is to construct their probabilistic copies and make them
jointly distributed. For instance, to find out how different are two
stochastically unrelated variables $X_{1}$ and $X_{2}$ one can consider
all jointly distributed variables $\left(Y_{1},Y_{2}\right)$ such
that $Y_{1}\mathop{=}\limits ^{d}X_{1}$ and $Y_{2}\mathop{=}\limits ^{d}X_{2}$,
and ask what is the minimal probability with which $Y_{1}$ and $Y_{2}$
can differ. Note that this question is meaningless if posed for the
$\left\{ X_{1},X_{2}\right\} $ themselves, but the minimal probability
of $Y_{1}\not=Y_{2}$ can be viewed as a degree of difference between
$X_{1}$ and $X_{2}$. The pair $\left(Y_{1},Y_{2}\right)$ in this
example is a special case of the notion defined next.
\begin{defn}
A \emph{coupling} of a set $\{X_{i}:i\in I\}$ of random variables
(generally stochastically unrelated) is a \emph{jointly distributed}
set $\{Y_{i}:i\in I\}$ of correspondingly indexed random variables
where each $Y_{i}$ has the same distribution as $X_{i}:$
\begin{equation}
Y_{i}\mathop{=}\limits ^{d}X_{i},
\end{equation}
for all $i\in I.$ Two couplings of the same set of random variables
are considered indistinguishable if they have the same distribution.\footnote{This means that the choice of a domain space for $\{Y_{i}:i\in I\}$
is irrelevant. There is a canonical way of constructing this space.
Let the set of values of $X_{i}$ be $E_{i}$, with the induced sigma-algebra
$\Sigma_{q}$. Then the domain space for $\{Y_{i}:i\in I\}$ can be
chosen as the set $\prod_{i\in I}E_{i}$ endowed with $\bigotimes_{i\in I}\Sigma_{i}$.
With this choice, every $Y_{i}$ is a coordinate projection function.}
\end{defn}

For the wealth of uses of this notion in probability theory, see e.g.
\cite{Thorisson2000}. The following special types of couplings are
important in analyzing systems of random variables. An \emph{independent
coupling} is a coupling $\{Y_{i}:i\in I\}$ such that all random variables
in it are independent. This coupling exists and is unique for any
$\{X_{i}:i\in I\}$. A \emph{maximal coupling} $\{Y_{1},Y_{2}\}$
of a pair $\{X_{1},X_{2}\}$ of random variables is a coupling that
maximizes the \emph{coupling probability} $\Pr\left[Y_{1}=Y_{2}\right]$
among all couplings of $\{X_{1},X_{2}\}$. This coupling also always
exist, but is not generally unique unless $X_{1},X_{2}$ are binary
variables.\footnote{Denoting the measurable spaces in which $X_{1},X_{2}$ are taking
their values by $\left(E_{1},\Sigma_{1}\right)$ and $\left(E_{2},\Sigma_{2}\right)$,
the definition of maximal coupling is predicated on the assumption
that the diagonal set $\left\{ \left(x,y\right)\in E_{1}\times E_{2}:x=y\right\} $
is measurable in $\Sigma_{1}\otimes\Sigma_{2}$. For dichotomous random
variables this condition is satisfied trivially.}
\begin{defn}
A \emph{multimaximal} coupling $\{Y_{i}:i\in I\}$ of $\{X_{i}:i\in I\}$
is a coupling such that $\{Y_{i},Y_{j}\}$ is a maximal coupling of
$\{X_{i},X_{j}\}$ for all $i,j\in I$.

Multimaximal couplings play a central role in contextuality analysis.
As a special case, the \emph{identity coupling} of $\{X_{i}:i\in I\}$
is a coupling $\{Y_{i}:i\in I\}$ such that all random variables in
it are identical: $\Pr\left[Y_{1}=Y_{2}\right]=1$ for all $i,j\in I$.
Such a coupling only exists if variables in $\{X_{i}:i\in I\}$ are
identically distributed, and then $\{Y_{i}:i\in I\}$ is the unique
multimaximal coupling of $\{X_{i}:i\in I\}$.
\end{defn}

The following theorem characterizes the multimaximal couplings for
binary (0/1) random variables.
\begin{thm}
\label{thm:multimax-equi}Let $\{X_{i}:i\in I\}$ be a set of binary
random variables, and let $\{Y_{i}:i\in I\}$ be a coupling of it.
Then, the following statements are equivalent:
\begin{enumerate}
\item The coupling $\{Y_{i}:i\in I\}$ is multimaximal.
\item \label{multimax-equi-finite-dist}Given any finite subset $\{i_{1},\dots,i_{n}\}\subset I$
of indices such that 
\[
\Pr[X_{i_{1}}=1]\ge\cdots\ge\Pr[X_{i_{n}}=1],
\]
the distribution of $Y_{i_{1}},\dots,Y_{i_{n}}$ is given by the probability
mass function $p(y_{i_{1}},y_{i_{2}},y_{i_{3}},\dots,y_{i_{n-1}},y_{i_{n}})$
determined by the $n+1$ probabilities (all other probabilities being
zero) 
\begin{align*}
{\textstyle p}(0,0,0,\dots,0,0) & =1-\Pr[X_{i_{1}}=1],\\
{\textstyle p}(1,0,0,\dots,0,0) & =\Pr[X_{i_{1}}=1]-\Pr[X_{i_{2}}=1],\\
{\textstyle p}(1,1,0,\dots,0,0) & =\Pr[X_{i_{2}}=1]-\Pr[X_{i_{3}}=1],\\
 & \vdots\\
{\textstyle p}(1,1,1,\dots,1,0) & =\Pr[X_{i_{n-1}}=1]-\Pr[X_{i_{n}}=1],\\
{\textstyle p}(1,1,1,\dots,1,1) & =\Pr[X_{i_{n}}=1].
\end{align*}
\item For any finite subset $\{i_{1},\dots,i_{n}\}\subset I$ of indices,
both 
\[
\Pr[Y_{i_{1}}=Y_{i_{2}}=\dots=Y_{i_{n}}=0]\textnormal{ and }\Pr[Y_{i_{1}}=Y_{i_{2}}=\dots=Y_{i_{n}}=1]
\]
are maximal possible probabilities among all couplings $\{Y_{i}:i\in I\}$
of $\{X_{i}:i\in I\}$.
\item For any pair of indices $i,j\in I$ such that $\Pr[X{}_{i}=1]\geq\Pr[X_{j}=1]$,
any one of the following statements:
\begin{enumerate}
\item \label{multimax-equi-11}$\Pr[Y{}_{i}=1,Y_{j}=1]=\Pr[X_{j}=1]$, the
maximal value among all couplings of $\left\{ X_{i},X_{j}\right\} $.
\item $\Pr[Y{}_{i}=0,Y_{j}=0]=\Pr[X_{i}=0]$, the maximal value among all
couplings of $\left\{ X_{i},X_{j}\right\} $.
\item \label{enu:multimax-equi-10}$\Pr[Y{}_{i}=0,Y_{j}=1]=0.$
\end{enumerate}
\end{enumerate}
\end{thm}

\begin{proof}
With reference to \cite{DzhCerKuj2017}, 1$\Rightarrow2\Rightarrow3$
$\Rightarrow$1. That any of 4a, 4b, 4c implies the other two is established
by direct computation, any of them is implied by 2, on putting $n=2$,
and 4a with 4b imply 1.
\end{proof}
\begin{thm}
\label{thm:unique}The multimaximal coupling $\{Y_{i}:i\in I\}$ of
a set of binary random variables $\{X_{i}:i\in I\}$ is unique.
\end{thm}

\begin{proof}
Consider two couplings $\{Y{}_{i}:i\in I\}$ and $\{Y'_{i}:i\in I\}$
of $\{X{}_{i}:i\in I\}$, and let $E_{i}$ denote the set of values
of $X_{i}$, for $i\in I$. It follows from Theorem \ref{thm:multimax-equi}
(statement \ref{multimax-equi-finite-dist}) that their distributions
agree on all \emph{cylinder subsets} of $E^{I}$ (i.e., Cartesian
products of measurable $S_{i}\subset E_{i}$, $i\in I$, with $S_{i}=E_{i}$
for all but a finite number of $i\in I$). The cylinder sets form
a \emph{$\pi$-system} (a nonempty collection of sets closed under
finite intersections). Since the two distributions agree on a $\pi$-system
(because cylinder sets correspond to finite subsets $J\subset I$),
it follows \cite{Williams1991} that they must agree on the $\sigma$-algebra
generated by the $\pi$-system, in this case the product $\sigma$-algebra
of $E^{I}$.
\end{proof}
\begin{thm}
\label{thm:pairwise-maximal-construction}Given any set of binary
random variables $\{X_{i}:i\in I\}$, the set $\{Y_{i}:i\in I\}$
defined by
\[
Y_{i}=[U\le\Pr(X_{i}=1)]:=\begin{cases}
1, & U\le\Pr(X_{i}=1),\\
0, & \text{otherwise,}
\end{cases}
\]
where $U$ is a uniform random variable on $[0,1]$, is the (unique)
multimaximal coupling of $\{X_{i}:i\in I\}$.
\end{thm}

\begin{proof}
That $\{Y_{i}:i\in I\}$ is a coupling of $\{X_{i}:i\in I\}$ follows
from the fact that (i) all $Y_{i}$ are functions of the same random
variable $U$ (hence they are jointly distributed), and (ii) $\Pr\left[Y_{i}=1\right]=\Pr[U\le\Pr(X_{i}=1)]=\Pr(X_{i}=1)$.
For any $i,j\in I$,
\[
\Pr\left[Y_{i}=Y_{j}=1\right]=\Pr\left[U\le\min\left[\Pr(X_{i}=1),\Pr(X_{j}=1)\right]\right]=\min\left[\Pr(X_{i}=1),\Pr(X_{j}=1)\right],
\]
which implies multimaximality by Theorem \ref{thm:multimax-equi}
(statement \ref{multimax-equi-11}).
\end{proof}

\section{\label{sec:Contextuality}Contextuality}
\begin{defn}
A \emph{coupling} of a system $\mathcal{R}=\{R_{q}^{c}:c\in C,\,q\in Q,\,q\Yleft c\}$
is an identically indexed set $S=\{S_{q}^{c}:c\in C,\,q\in Q,\,q\Yleft c\}$
of random variables whose set of bunches $\{S^{c}:c\in C\}$ is a
coupling of the set of bunches $\{R^{c}:c\in C\}$ of $\mathcal{R}$.
\end{defn}

In other words, $S$ is a coupling of $\mathcal{R}$ if (as suggested
by the notation) the elements of $S$ are jointly distributed and,
for any $c\in C$,
\begin{equation}
S^{c}\overset{d}{=}R^{c}.
\end{equation}

Clearly, $S$ contains as its marginals the couplings $S_{q}=\left\{ S_{q}^{c}:c\in C,q\Yleft c\right\} $
for each of the connections $\mathcal{R}_{q}=\left\{ R_{q}^{c}:c\in C,q\Yleft c\right\} $
of $\mathcal{R}$. We can view the coupling $S$ as a system in its
own right, and its marginals $S_{q}$ as connections of this system.
Then we can equivalently say either that, within $S$, connections
$\mathcal{R}_{q}$ have couplings $S_{q}$ with some property $P$
(e.g., multimaximal couplings) or that $\mathcal{R}$ has a coupling
$S$ whose connections $S_{q}$ have the property $P$ (e.g., multimaximal
connections).

The traditional approach \cite{BudroniBook2016,Cabello2013,AbramBarbMans2011,AbramskyBrand2011,Araujoetal2013,Kurzynski2014,Kurzynskietal2012,Fine1982,SuppesZanotti1981}
is that a system is noncontextual if it has a coupling whose connection
$S_{q}$ is the identity coupling of $\mathcal{R}_{q}$, for every
$q\in Q$. Recall that in the identity coupling, $\Pr\left[S_{q}^{c}=S_{q}^{c'}\right]=1$
for all components of the connection $S_{q}$.\footnote{The notion of a coupling in the traditional approach is not used explicitly
(see \cite{DzhKujFoundations2017,Dzh2019} for difficulties this creates).
To our knowledge, Thorisson \cite{Thorisson2000} (Ch.$\:$1, Sec.$\:$10.4,
p.$\:$29) was first to use couplings in contextuality analysis of
a system. In CbD, they play a central role.} Thus, for a system $\mathcal{R}$ to be noncontextual in the traditional
sense, the system must be consistently connected, and all random variables
$R_{q}^{c}$, $R_{q}^{c'}$, $R_{q}^{c''}$, $\dots$ in every connection
$\mathcal{R}_{q}$ must correspond to one and the same random variable
$T_{q}:=S_{q}^{c}=S_{q}^{c'}=S_{q}^{c''}=\dots$ in some coupling
$S$ of $\mathcal{R}$. If a system is consistently connected but
such a coupling does not exist, the system is considered contextual.
For an inconsistently connected system, identity couplings of connections
do not exist, because $\Pr\left[S_{q}^{c}=S_{q}^{c'}\right]$ cannot
reach $1$ unless $R_{q}^{c}\overset{d}{=}R_{q}^{c'}$. All inconsistently
connected systems therefore, if one follows the logic of the traditional
approach, have to be treated as (trivially) contextual, or else as
systems whose contextuality status is undefined. This violates a general
principle that we will now formulate.

Let us define the format of the system $\mathcal{R}$ in (\ref{eq:system})
as 
\begin{equation}
\mathsf{f}:=\left(\Yleft,\{\left(E_{q},\Sigma_{q}\right):q\in Q\}\right).
\end{equation}
A format therefore specifies which content $q$ is measured in which
context $c$ (the sets $Q$ and $C$ are then effectively determined
as projections of $\Yleft$), and it also specifies the type of the
random variables involved, i.e. their sets $E_{q}$ of possible values
and the associated sets $\Sigma_{q}$ of events. Let us agree to exclude
the trivial formats in which every connection consists of a single
random variable. The principle in question is as follows.
\begin{description}
\item [{Analyticity}] For any given format, among all inconsistently connected
systems of this format there are noncontextual systems.
\end{description}
In other words, contextuality status of an inconsistently connected
system should depend on its bunch distributions rather than be predetermined
by its format. The importance of this principle is that it rules out
trivial extensions of the traditional contextuality theory, including
the one that declares all inconsistently connected systems contextual.

In CbD, the concept of (non-)contextuality is extended to inconsistently
connected systems by replacing identity couplings with multimaximal
couplings: the general idea is that a system is noncontextual if,
for all $q\Yleft c,c'$ simultaneously, the value of $\Pr\left[S_{q}^{c}=S_{q}^{c'}\right]$
in some coupling $S$ reaches its maximum (which is $1$ if and only
if $R_{q}^{c}\overset{d}{=}R_{q}^{c'}$). However, the established
definition in CbD requires that the variables in the systems be dichotomized
prior to being subjected to contextuality analysis.
\begin{defn}
\label{def:system dichotomization}A \emph{split representation }of
a system $\mathcal{R}=\{R_{q}^{c}:c\in C,q\in Q,q\Yleft c\}$ is a
system 
\[
\mathcal{D}=\{R_{q,A}^{c}:c\in C,q\in Q,A\in\varUpsilon_{q},\left(q,A\right)\Yleft c\},
\]
where

(i) $R_{q,A}^{c}$ are binary variables defined by (\ref{eq:dichtomization}),

(ii) the values of $R_{q,A}^{c}$ for all $A\in\varUpsilon_{q}$ uniquely
determine the value of $R_{q}^{c}$,

(iii) $\varUpsilon_{q}\subseteq\Sigma_{q}$, and $\Sigma_{q}$ is
the minimal sigma-algebra containing $\varUpsilon_{q}$,

(iv) $\left(q,A\right)\Yleft c$ if and only if $q\Yleft c$ and $A\in\varUpsilon_{q}$.\footnote{There is a slight abuse of notation here: we use the same symbol $\Yleft$
to indicate the format relation of both $\mathcal{R}$ and $\mathcal{D}$.}
\end{defn}

For any $q\in Q$, the subsystem 
\[
\mathcal{D}_{q}=\{R_{q,A}^{c}:c\in C,A\in\varUpsilon_{q},\left(q,A\right)\Yleft c\}
\]
 is a split-representation of the system consisting of the single
connection $\mathcal{R}_{q}$.

The indexation of $\varUpsilon_{q}$ implies that the same set of
dichotomizations is applied to all random variables in a given connection
$\mathcal{R}_{q}$. All these variables, we remind, have the same
set of values $E_{q}$ and the same sigma-algebra $\Sigma_{q}$. We
also remind the convention (to avoid trivial redundancy) that if $A\in\varUpsilon_{q}$
then $E_{q}-A\not\in\varUpsilon_{q}$, and $A$ is a proper, nonempty
subset of $E_{q}$.

The split representation $\mathcal{D}$ retains the same set of contexts
$C$ as in $\mathcal{R}$ but splits each ``old'' content $q$ into
a set of ``new'' contents $\left\{ \left(q,A\right):A\in\varUpsilon_{q}\right\} $.
For example, suppose that in the system
\begin{equation}
\begin{array}{|c|c|c|c}
\hline R_{1}^{1} & R_{2}^{1} &  & c=1\\
\hline R_{1}^{2} &  & R_{3}^{2} & c=2\\
\hline  & R_{2}^{3} & R_{3}^{3} & c=3\\
\hline q=1 & q=2 & q=3 & \mathcal{R}
\end{array}
\end{equation}
the random variables in connection $\mathcal{R}_{q=1}$ have values
$E_{1}=\{1,2,3,4\}$, the variables in connection $\mathcal{R}_{q=2}$
have values $E_{2}=\{a,b,c\}$, and the variables in connection $\mathcal{R}_{q=3}$
are binary. Then, we could represent the original system as
\begin{equation}
\begin{array}{|c|c|c|c|c|c|c}
\hline \left[R_{1}^{1}\in\{1\}\right] & \left[R_{1}^{1}\in\{1,2\}\right] & \left[R_{1}^{1}\in\{1,2,3\}\right] & \left[R_{2}^{1}\in\{a\}\right] & \left[R_{2}^{1}\in\{b\}\right] &  & c=1\\
\hline \left[R_{1}^{2}\in\{1\}\right] & \left[R_{1}^{2}\in\{1,2\}\right] & \left[R_{1}^{2}\in\{1,2,3\}\right] &  &  & R_{3}^{2} & c=2\\
\hline  &  &  & \left[R_{2}^{3}\in\{a\}\right] & \left[R_{2}^{3}\in\{b\}\right] & R_{3}^{3} & c=3\\
\hline q=\left(1,\left\{ 1\right\} \right) & q=\left(1,\left\{ 1,2\right\} \right) & q=\left(1,\left\{ 1,2,3\right\} \right) & q=\left(2,\left\{ a\right\} \right) & q=\left(2,\left\{ b\right\} \right) & q=3 & \mathcal{D}
\end{array}.
\end{equation}

Since the choice of a split representation of a given $\mathcal{R}$
is not unique without additional constraining principles (to be discussed
later), we denote 
\begin{equation}
\varUpsilon:=\left\{ \varUpsilon_{q}:q\in Q\right\} ,
\end{equation}
and refer to $\mathcal{D}$ in Definition \ref{def:system dichotomization}
as the $\varUpsilon$-split representation of $\mathcal{R}$.
\begin{defn}
\label{def:binary-noncontextuality}A system $\mathcal{R}$ is $\varUpsilon$-\emph{noncontextual}
if its $\varUpsilon$-split representation $\mathcal{D}$ has a coupling
whose connections are multimaximal (such a coupling is called \emph{multimaximally
connected}). Otherwise $\mathcal{R}$ is $\varUpsilon$-\emph{contextual.}
\end{defn}

\noindent Recall that the connections of a coupling of $\mathcal{D}$
are multimaximal if and only if they are multimaximal couplings of
the connections of $\mathcal{D}$. Obviously, $\varUpsilon$-split
representation of a system is a system of binary random variables,
and it is its own and only split representation (up to relabeling
of values).

We will see in the following that at least in some cases the set $\varUpsilon$
need not be mentioned because its choice is determined uniquely by
certain principles, to be formulated in Section \ref{sec:General-principles}.
Even without these principles, however, $\varUpsilon$ obviously need
not be mentioned if the variables in $\mathcal{R}$ are binary to
begin with. If the sets of contents and contexts in such a system
are finite, the contextuality status of the system (as well as measures
of (non)contextuality, not discussed in this paper) can be computed
by linear programming. For the important special case of cyclic systems,
the contextuality status and measures of (non)contextuality can be
determined analytically based on formulas derived in \cite{KujDzhProof2016,DKC2020,KujDzhMeasures}.
A cyclic system of rank $n\in\{2,3,\dots\}$ has $2n$ binary random
variables arranged in bunches $\{R_{i}^{i},R_{i\oplus1}^{i}\}$ for
$i=1,\dots,n$, where $i\oplus1=i+1$ for $i<n$ and $n\oplus1=1$.
Thus, (\ref{eq:1}) and (\ref{eq:order system}) represent cyclic
systems of ranks 4 and 2, respectively. Cyclic systems cover a large
part of traditionally considered systems in physics and behavioral
psychology. CbD allows one to analyze these systems with any amount
of inconsistent connectedness present.

Definition \ref{def:binary-noncontextuality} can be equivalently
stated as follows:
\begin{defn}
A system $\mathcal{R}$ is $\varUpsilon$-\emph{noncontextual} if
it has a coupling $S$ whose $\varUpsilon$-split representation is
multimaximally connected.
\end{defn}

\noindent This version is often preferable, because of the following
observation.
\begin{lem}
A coupling $S$ of a system $\mathcal{R}$, and the $\varUpsilon$-split
representation of $S$ uniquely determine each other.
\end{lem}

\begin{proof}
Given a coupling $S$ (which is a system in its own right), a construction
of its $\varUpsilon$-split representation is given by Definition~\ref{def:system dichotomization}(i).
Conversely, given a $\varUpsilon$-split representation $D$ of a
coupling $S$, Definition~\ref{def:system dichotomization}(ii) implies
that each random variable $S_{c}^{q}$ of $S$ is fully determined
by its representation as $\{R_{q,A}^{c}:A\in\varUpsilon_{q}\}$ in
$D$, and so the joint distribution of all $S_{q}^{c}$ is also determined
as all random variables in $D$ are jointly distributed.
\end{proof}
It is clear from the proof that even though the split representation
$\mathcal{D}$ may be very large, the support of its coupling has
the same cardinality as the support of a coupling $S$ of the original
system $\mathcal{R}$.

We conclude this section with the following simple observation.
\begin{thm}
The definition of contextuality in CbD satisfies Analyticity with
respect to any $\varUpsilon$-split representation.
\end{thm}

\begin{proof}
Given any format, choose, e.g., $\mathcal{R}$ of this format in which
all random variables are deterministic (making sure their values vary
within a connection so that the system is inconsistently connected).
A unique coupling $S$ then trivially exist for $\mathcal{R}$, and
any connection in this coupling is multimaximal. Any split representation
of $S$ will retain the multimaximality of connections, proving that
$\mathcal{R}$ is $\varUpsilon$-noncontextual with respect to any
$\varUpsilon$.
\end{proof}

\section{\label{sec:General-principles}Why multimaximality, and why dichotomizations?}

CbD is an extension of the traditional understanding of contextuality
effected by two modifications thereof: (1) the replacement of identity
couplings of connections with multimaximal couplings, and (2) the
replacement of systems of random variables with their split representations.
Both these modifications have been justified in previous CbD publications
\cite{DzhCerKuj2017}. We will recapitulate these justifications briefly.

The only alternative to multimaximal coupling proposed in the literature
as a generalization of identity couplings is the notion of a globally
maximal coupling. Such a coupling maximizes the probability of 
\begin{equation}
S_{q}^{c_{1}}=S_{q}^{c_{2}}=\ldots=S_{q}^{c_{k}},
\end{equation}
where $\left\{ c_{1},\ldots,c_{n}\right\} $ are all contexts in which
$q$ is measured (i.e., $\left\{ S_{q}^{c_{1}},\ldots,S_{q}^{c_{k}}\right\} $
is a coupling of an entire connection). This generalization was adopted
in an earlier version of CbD, but abandoned later as it fails to satisfy
the following principle.
\begin{description}
\item [{Noncontextual$\:$Nestedness}] Any subsystem of a noncontextual
system is noncontextual.
\end{description}
A subsystem of a system is created by removing certain random variables
from the system. Consider, e.g., the system
\begin{equation}
\begin{array}{|c|c|c}
\hline R_{1}^{1} & R_{2}^{1} & c=1\\
\hline R_{1}^{2} & R_{2}^{2} & c=2\\
\hline R_{1}^{3} & R_{2}^{3} & c=3\\
\hline R_{1}^{4} & R_{2}^{4} & c=4\\
\hline q=1 & q=2 & \mathcal{R}
\end{array}
\end{equation}
with binary random variables whose bunches are distributed as
\begin{equation}
\begin{array}{ccc}
\begin{array}{c|c|c}
 & R_{2}^{1}=1 & R_{2}^{1}=0\\
\hline R_{1}^{1}=1 & \nicefrac{1}{2} & 0\\
\hline R_{1}^{1}=0 & 0 & \nicefrac{1}{2}
\end{array} &  & \begin{array}{c|c|c}
 & R_{2}^{2}=1 & R_{2}^{2}=0\\
\hline R_{1}^{2}=1 & 0 & \nicefrac{1}{2}\\
\hline R_{1}^{2}=0 & \nicefrac{1}{2} & 0
\end{array}\\
\\
\begin{array}{c|c|c}
 & R_{2}^{3}=1 & R_{2}^{3}=0\\
\hline R_{1}^{3}=1 & 1 & 0\\
\hline R_{1}^{3}=0 & 0 & 0
\end{array} &  & \begin{array}{c|c|c}
 & R_{2}^{4}=1 & R_{2}^{4}=0\\
\hline R_{1}^{4}=1 & 0 & 0\\
\hline R_{1}^{4}=0 & 0 & 1
\end{array}
\end{array}.
\end{equation}
It is easy to show that the maximal probability of $S_{q}^{1}=S_{q}^{2}=S_{q}^{3}=S_{q}^{4}$
is zero for both $q=1$ and $q=2$. Consequently, any coupling of
$\mathcal{R}$ has globally maximal connections. If we adopt the definition
of contextuality based on globally maximal couplings, then this system
is noncontextual. At the same time, the subsystem 
\begin{equation}
\begin{array}{|c|c|c}
\hline R_{1}^{1} & R_{2}^{1} & c=1\\
\hline R_{1}^{2} & R_{2}^{2} & c=2\\
\hline q=1 & q=2 & \mathcal{R}'\subset\mathcal{R}
\end{array}
\end{equation}
is, by the same definition (which in this case coincides with our
Definition \ref{def:binary-noncontextuality}), contextual. In fact,
it has the maximal degree of contextuality among all cyclic systems
of rank 2 \cite{DKC2020}. This example also shows that the use of
globally maximal connections violates another reasonable principle,
formulated next.
\begin{description}
\item [{Deterministic$\:$Redundancy}] Any deterministic random variable
can be deleted from a system without affecting its contextuality status;
and for any $\left(q,c\right)\not\in\Yleft$ ($q$ is not measured
in $c$), one can add a deterministic $R_{q}^{c}$ without affecting
the system's contextuality status.
\end{description}
As we use multimaximality to define contextuality, the Noncontextual
Nestedness principle is satisfied trivially, and it is shown in \cite{Dzh2017Nothing}
that the Deterministic Redundancy principle is satisfied as well.
However, these and other constraints stipulated in this paper do not
determine multimaximality uniquely. A generalization of multimaximally
connected couplings, dubbed \emph{$\mathsf{C}$-couplings}, has in
fact been considered \cite{DzhKuj2.0,Dzh2017Nothing}, in which maximality
of $\Pr\left[S_{q}^{c},S_{q}^{c'}\right]$ is replaced by an arbitrary
property \emph{$\mathsf{C}$} that every pair $\left(S_{q}^{c},S_{q}^{c'}\right)$
in a $\mathsf{C}$-coupling $S$ has to satisfy. Any $\mathsf{C}$-coupling
satisfies the principles of Noncontextual Nestedness and Deterministic
Redundancy. At present, however, we do not know reasonable alternatives
to the maximality of $\Pr\left[S_{q}^{c},S_{q}^{c'}\right]$ as a
realization of property $\mathsf{C}$.

The main reason why CbD requires dichotomization is that outside the
class of binary random variables the notion of contextuality does
not satisfy the following principle.
\begin{description}
\item [{Coarse-graining}] A noncontextual system remains noncontextual
following coarse-graining of its random variables.
\end{description}
A coarse-graining is a measurable function $f_{q}:E_{q}\to E'_{q}$,
for $q\in Q$. Thus, a coarse-graining maps $R_{q}^{c}$ into another
random variable $f_{q}\left(R_{q}^{c}\right)$ by lumping together
certain elements of $E_{q}$ (the set of values of $R_{q}^{c}$),
doing this in the same way for all variables in a connection. Dichotomization
is a special case of coarse-graining, with $E_{q}'=\left\{ 0,1\right\} $.

Consider the following system:
\begin{equation}
\begin{array}{|c|c|c}
\hline R_{1}^{1} & R_{2}^{1} & c=1\\
\hline R_{1}^{2} & R_{2}^{2} & c=2\\
\hline q=1 & q=2 & \mathcal{R}
\end{array},
\end{equation}
with the bunches distributed as
\begin{equation}
\begin{array}{ccc}
\begin{array}{c|c|c|c|c}
 & R_{2}^{1}=1 & R_{2}^{1}=2 & R_{2}^{1}=3 & R_{2}^{1}=4\\
\hline R_{1}^{1}=1 & \nicefrac{1}{2} & 0 & 0 & 0\\
\hline R_{1}^{1}=2 & 0 & 0 & 0 & 0\\
\hline R_{1}^{1}=3 & 0 & 0 & \nicefrac{1}{2} & 0\\
\hline R_{1}^{1}=4 & 0 & 0 & 0 & 0
\end{array} &  & \begin{array}{c|c|c|c|c}
 & R_{2}^{2}=1 & R_{2}^{2}=2 & R_{2}^{2}=3 & R_{2}^{2}=4\\
\hline R_{1}^{2}=1 & 0 & 0 & 0 & 0\\
\hline R_{1}^{2}=2 & 0 & 0 & 0 & \nicefrac{1}{2}\\
\hline R_{1}^{2}=3 & 0 & 0 & 0 & 0\\
\hline R_{1}^{2}=4 & 0 & \nicefrac{1}{2} & 0 & 0
\end{array}\end{array}.
\end{equation}
This system is noncontextual, because any coupling thereof has (multi)maximal
connections. However, if we coarse-grain (here, dichotomize) them
by 
\begin{equation}
f_{1}:\;\downarrow\begin{array}{cccc}
1 & 2 & 3 & 4\\
1 & 1 & 0 & 0
\end{array},f_{2}:\;\downarrow\begin{array}{cccc}
1 & 2 & 3 & 4\\
1 & 1 & 0 & 0
\end{array},
\end{equation}
the bunches of the new system will be distributed as
\begin{equation}
\begin{array}{ccc}
\begin{array}{c|c|c}
 & R_{2}^{1}=1 & R_{2}^{1}=0\\
\hline R_{1}^{1}=1 & \nicefrac{1}{2} & 0\\
\hline R_{1}^{1}=0 & 0 & \nicefrac{1}{2}
\end{array} &  & \begin{array}{c|c|c}
 & R_{2}^{2}=1 & R_{2}^{2}=0\\
\hline R_{1}^{2}=1 & 0 & \nicefrac{1}{2}\\
\hline R_{1}^{2}=0 & \nicefrac{1}{2} & 0
\end{array}\end{array}\,,
\end{equation}
and this system, as already stated in the previous example, is contextual.

By contrast, if a system consists of binary random variables, the
Coarse-graining principle is satisfied trivially. A coarse-graining
of a binary variable maps it into itself (modulo renaming its values)
or into a deterministic variable, attaining a single value with probability
1. The latter cannot violate the Coarse-graining principle because
the CbD definition of contextuality satisfies the Deterministic Redundancy
principle.

If the system to be analyzed is consistently connected, multimaximal
couplings reduce to identity couplings, and (non)contextuality in
CbD properly specializes to the traditional understanding of (non)contextuality
(provided the latter is rigorously stated in terms of couplings).
The choice of a split representation for a consistently connected
system is inconsequential, and may even be omitted as a matter of
convenience.

\section{\label{sec:V}How to choose dichotomizations?}

The intuition behind how one has to do coarse-graining in general
and dichotomization in particular is simple: allowable ``lumping''
should only lump together ``contiguous'' sets of values. This intuition
is captured by the notion of (pre-topologically, or V-) \emph{linked
sets}.\footnote{This is essentially a weak form of pre-topological connectedness,
but we avoid using the latter word to prevent confusing it with its
use in CbD, in such terms as ``multimaximally connected'' or ``consistently
connected,'' derived from the term ``connection'' for the set of
random variables sharing a content.}

We define a \emph{symmetrical Fr\'echet V-space} (see \cite{Sierpinski1956}
for a general theory of Fr\'echet V-spaces) as a non-empty set $E$
endowed with a collection $\mathbb{V}$ of nonempty subsets of $E$,
called \emph{vicinities}. A vicinity $V$ can be called a vicinity
of any element of $V$, and every element of $E$ has to have a vicinity.
The term ``symmetrical'' reflects the fact that if $y$ is in a
vicinity of $x$, then $x$ is in the vicinity of $y.$ A topological
space is a symmetrical V-space with additional properties that we
do not need to use.

Let us illustrate this and related concepts on a simple example. In
a psychophysical experiment described in \cite{CervDzh2019}, a small
visual object (a ``dot'') could be in one of five positions, as
shown, 
\begin{equation}
\boxed{\begin{array}{ccccc}
 &  & *\\
* &  & * &  & *\\
 &  & *
\end{array}},\label{eq:5-points}
\end{equation}
and an observer had to identify the position as \emph{center}, \emph{left},
\emph{right}, \emph{up}, or \emph{down}. Thus the response of the
observer was a 5-valued random variable, and we take this set of 5
values as $E$. Let us associate to each point $x$ of $E$ as its
vicinities all sets consisting of $x$ and its one-step-away neighbor.
For instance, the point \emph{$left$ }has the vicinities $V_{1}=\left\{ left,up\right\} $,
$V_{2}=\left\{ left,center\right\} $, and $V_{3}=\left\{ left,down\right\} $.

To define V-linked sets, we need to remind the concept of \emph{limit
points} (generalized to V-spaces). Given a V-space $E$, a point $x\in E$
is a limit point of a set $F\subseteq E$ if every vicinity of $x$
contains a point of $F$ other than $x$.
\begin{defn}
\label{def:V-linked}A subset $F$ of a V-space $E$ is V-linked if

(i) $F$ is a singleton or a vicinity of some point in $E$;

(ii) $F$ is a union of a V-linked set and a subset of its limit points;

(iii) $F$ is a union of V-linked sets with a nonempty intersection.
\end{defn}

When dealing with a random variable $R_{q}^{c}$ whose set of values
$E_{q}$ is endowed with a sigma algebra $\Sigma_{q}$, the latter
does not generally determine the choice of a V-space $\mathbb{V}_{q}$
for $E_{q}$ uniquely (and vice versa). However, in ``ordinary''
cases, we have a natural choice of $\mathbb{V}_{q}$ and a natural
choice of $\Sigma_{q}$ for $E_{q}$ that satisfy the following definition.
\begin{defn}
Let a random variable $R$ have a set of possible values $E$ endowed
with a V-space $\mathbb{V}$ and a sigma-algebra $\Sigma$. The variable
$R$ is said to be \emph{ordinary} if $\Sigma$ is the smallest sigma-algebra
containing all the vicinities in $\mathbb{V}$.
\end{defn}

In our example (\ref{eq:5-points}), one can check that the smallest
sigma-algebra containing all the vicinities is the power set of $E$
(because every singleton can be obtained by appropriate intersections
of the vicinities).

We are ready now to stipulate the definition that guides our choice
of dichotomizations.
\begin{defn}
\label{def:allowable}An allowable coarse-graining of V-space $E$
is a surjection $f:E\rightarrow E'$ such that $E'$ is a V-space,
and

(i) for any V-linked subset $X$ of $E$, $f\left(X\right)$ is a
V-linked subset of $E'$, and

(ii) for any V-linked subset $Y$ of $E'$, $f^{-1}\left(Y\right)$
is a V-linked subset of $E$.

A dichotomization is a mapping $f:E\rightarrow E'$ where the vicinities
in $E'=\left\{ 0,1\right\} $ are taken to be $\left\{ 0\right\} $,
$\left\{ 1\right\} $, and $\left\{ 0,1\right\} $. So for any $E$,
the dichotomization is allowable if and only if $D_{0}=f^{-1}\left(0\right)$,
$D_{1}=F^{-1}(1)$, and $E=F^{-1}(\{0,1\})$ are V-linked.
\end{defn}

Thus, in our example (\ref{eq:5-points}), there are 15 distinct partitions
of the set into two subsets, and all of them are allowable except
for
\begin{equation}
\boxed{\begin{array}{ccccc}
 &  & \circ\\
\bullet &  & \circ &  & \bullet\\
 &  & \circ
\end{array}},\boxed{\begin{array}{ccccc}
 &  & \bullet\\
\circ &  & \circ &  & \circ\\
 &  & \bullet
\end{array}},
\end{equation}
where the filled circles form non-linked sets $D_{0}$.

The proof of the following statement is obvious.
\begin{thm}
Allowable coarse-grainings are closed under compositions, that is,
if $f:E\to E'$ and $g:E'\to E''$ are allowable coarse-grainings,
then $g\circ f:E\to E''$ is an allowable coarse-graining. In particular,
every allowable dichotomization $d:E'\to\{0,1\}$ of an allowably
coarse-grained V-space $E'=f(E)$ yields an allowable dichotomization
$d\circ f:E\to\{0,1\}$ of the original space $E$.
\end{thm}

It follows that if one forms the split representation of a given system
$\mathcal{R}$ by all allowable dichotomizations of each connection,
then the Coarse-graining principle is satisfied. Indeed, a split representation
of a coarse-grained system is merely a subsystem of the split representation
of the original system, because of which if the latter is noncontextual,
then so is the former.

In the case of random variables with linearly ordered sets of values
$E\subseteq\mathbb{R}$, the natural vicinities of $x\in E$ can be
chosen as all intervals $\left\{ z:a<z<b\right\} $ containing $x$,
and the natural sigma-algebra is the Borel sigma-algebra. The only
linked subsets of $E$ are intervals. Thence the allowable dichotomizations
are cuts:
\begin{equation}
\left\{ D_{0}\left(a\right)=\left\{ x:x\leq a\right\} ,D_{1}\left(a\right)=\left\{ x:x>a\right\} \right\} 
\end{equation}
and
\begin{equation}
\left\{ D'_{0}\left(a\right)=\left\{ x:x<a\right\} ,D'_{1}\left(a\right)=\left\{ x:x\geq a\right\} \right\} ,
\end{equation}
for all $a\in E$. One of these two types can be dropped if the other
is used, as shown in the next section.

In the case $E$ is a region of $\mathbb{R}^{n}$, the situation is
more complex, as one may associate with it many ``natural'' but
``uninteresting'' V-spaces, making too many types of dichotomizations
allowable. This leads to all inconsistently connected systems being
contextual, in violation of the Analyticity principle. However, a
variable with values in $\mathbb{R}^{n}$ can always be treated as
$n$ jointly distributed real-valued variables, in which case the
choice reduces to the one previously considered. At present, we do
not know whether there are other approaches to $\mathbb{R}^{n}$-valued
variables that comply with Analyticity.

In the case of a categorical random variable, $E=\left\{ 1,\ldots,r\right\} $,
its V-space involves all possible subsets, the same as its sigma-algebra.
Definition \ref{def:allowable} then allows for all possible dichotomizations.

There seems to be no need to multiply examples, as they are easily
construable.

\section{\label{sec:cuts}Cut dichotomizations of real-valued random variables}

Consider a \emph{single connection} $\left\{ R_{q}^{1},\dots,R_{q}^{n}\right\} $
of a system, with all random variables being defined on the set of
reals endowed with the usual (Borel) V-space and the Borel sigma-algebra.
This includes variables with continuous distribution function, but
also a variety of discrete linearly ordered random variables, such
as spin measurements in quantum physics, which have ordering $(-\frac{1}{2},\frac{1}{2})$
for spin-$\nicefrac{1}{2}$ particles, $(-1,0,1)$ for spin-$1$ particles,
$(-\frac{3}{2},-\frac{1}{2},\frac{1}{2},\frac{3}{2})$ for spin $\nicefrac{3}{2}$-particles,
etc. Thus, virtually all measurements in physics can be modeled using
real-valued random variables. 

Since the content $q$ is fixed, we can drop this subscript and present
our connection as $\left\{ R^{1},\dots,R^{n}\right\} $. In accordance
with Section \ref{sec:V}, we replace $\left\{ R^{1},\dots,R^{n}\right\} $
with the set of its cuts, that is, we form the system of binary random
variables
\begin{equation}
R_{(-\infty,x]}^{k}:=[R^{k}\le x]:=\begin{cases}
1 & \textnormal{if }R^{k}\le x,\\
0 & \text{otherwise,}
\end{cases}\label{eq:cuts}
\end{equation}
for $k=1,\dots,n$ (rows) and $x\in\mathbb{R}$ (columns). We could
also have chosen them as
\begin{equation}
R_{(-\infty,x)}^{k}:=[R^{k}<x]:=\begin{cases}
1 & \textnormal{if }R^{k}<x,\\
0 & \text{otherwise,}
\end{cases}\label{eq:othercuts}
\end{equation}
but it would not make any difference. Indeed, the set of points 
\begin{equation}
\mathbb{R}_{0}=\left\{ x\in\mathbb{R}:\Pr[R^{k}\le x]>\Pr[R^{k}<x]\textnormal{ for some }k\in\left\{ 1,\ldots,n\right\} \right\} 
\end{equation}
is at most countable. Let us indicate the elements of $\mathbb{R}-\mathbb{R}_{0}$
by $\bar{x}$. Consider any coupling $\left\{ S^{1},\ldots,S^{n}\right\} $
of $\left\{ R^{1},\dots,R^{n}\right\} $. Let $a\in\mathbb{R}$ and
$i,i'\in\left\{ 1,\dots,n\right\} $ be arbitrary. Using the right-continuity
of the distribution functions,
\begin{equation}
\Pr\left[S^{i}\leq a,S^{i'}\leq a\right]=\lim_{\bar{x}\to a^{+}}\Pr\left[S^{i}\leq\bar{x},S^{i'}\leq\bar{x}\right],
\end{equation}
for any $a\in\mathbb{R}$. 

But then
\begin{equation}
\Pr\left[S^{i}\leq\bar{x},S^{i'}\leq\bar{x}\right]=\min\left\{ \Pr\left[S^{i}\leq\bar{x}\right],\Pr\left[S^{i'}\leq\bar{x}\right]\right\} 
\end{equation}
for all $\bar{x}\in\mathbb{R}-\mathbb{R}_{0}$ implies
\begin{equation}
\Pr\left[S^{i}\leq a,S^{i'}\leq a\right]=\min\left\{ \Pr\left[S^{i}\leq a\right],\Pr\left[S^{i}\leq a\right]\right\} ,
\end{equation}
for all $a\in\mathbb{R}$. By Theorem~\ref{thm:multimax-equi} (statement
\ref{multimax-equi-11}), this means that if the cuts are defined
by (\ref{eq:cuts}), multimaximality of the split representation of
$\left\{ S^{1},\ldots,S^{n}\right\} $ is implied by the multimaximality
of the same split representation from which all connections corresponding
to $x\in\mathbb{R}_{0}$ are removed. Since the reverse implication
is trivial, we can replace the implication with equivalence. We can
analogously prove the same for the split representations defined by
(\ref{eq:othercuts}), using the left-continuity instead of the right-continuity.

Let us therefore choose (\ref{eq:cuts}) for subsequent analysis,
and let us write $R_{(-\infty,x]}^{k}$ more conveniently as $R_{x}^{k}$.
\begin{thm}
\label{thm:continuous-cuts}The split representation of a single connection
formed by cuts as given by (\ref{eq:cuts}) is noncontextual.
\end{thm}

\begin{proof}
This system has the coupling $\{S_{x}^{k}:k=1,\dots,n,\ x\in\mathbb{R}\}$
where
\[
S_{x}^{k}=[F_{k}^{-1}(U)\le x]=[U\le F_{k}(x)],
\]
$U$ is a $[0,1]$ uniform random variable, and $F_{k}$ and $F_{k}^{-1}$
are, respectively, the cumulative distribution function and the quantile
function of $R^{k}.$ As $F_{k}^{-1}(U)\mathop{=}\limits ^{d}R^{k}$,
the first equality implies that this is indeed a coupling of the system.
The second equality implies, by Theorem~\ref{thm:pairwise-maximal-construction},
that this coupling is multimaximal. As the whole coupling is multimaximal,
its connections are also multimaximal, and the system is noncontextual.
\end{proof}
This is quite a difference from considering all possible dichotomizations,
which leads to all inconsistently connected single connections to
be contextual \cite{DzhCerKuj2017}. Since a single connection can
be viewed as a system of random variables, generally inconsistently
connected, the theorem shows that split representations formed by
cuts satisfy the Analyticity principle. Of course, a system consisting
of more than one connection may very well be contextual.

Consider an arbitrary coupling $\left\{ S^{1},\dots,S^{n}\right\} $
of a single connection $\left\{ R^{1},\dots,R^{n}\right\} $ such
that the split representation of this coupling is multimaximally connected.
Theorem \ref{thm:continuous-cuts} says such couplings exist, but
the specific coupling constructed in this theorem is not the only
possible multimaximally connected coupling. We will now analyze the
constraints a coupling $\{S_{x}^{k}:k=1,\dots,n,\ x\in\mathbb{R}\}$
with multimaximal connections imposes on the joint distributions of
the coupling $\left\{ S^{1},\dots,S^{n}\right\} $. Let us choose
two arbitrary elements of the coupling and denote them $S^{i}$ and
$S^{j}$, and let $F_{i}$ and $F_{j}$ be their respective distribution
functions (i.e., the distribution functions of $R^{i}$ and $R^{j}$).

\begin{figure}
\begin{centering}
\begin{tikzpicture}[>=latex]
\draw[thick, ->] (0,0) -- node[at end,below] {$S^i$} (8.5,0);
\draw[thick, ->] (0,3) -- node[at end,left] {$S^j$} (0,8.5);
\draw[thick, dotted, fill=lightgray] (0,0) -- (0,3) -- (3,3) -- cycle;
\draw[thick] (0,3) -- (3,3);
\draw[thick, dotted, fill=lightgray] (3.5,3.5) -- (5,3.5) -- (5,5) -- cycle;
\draw[thick] (5,3.5) -- (5,5);
\draw[thick, dotted, fill=lightgray] (5,5) -- (5,7) -- (7,7) -- cycle;
\draw[thick] (5,7) -- (7,7);
\draw[thick, dotted, fill=lightgray] (7,7) -- (8,7) -- (8,8) -- cycle;
\draw[thick] (8,7) -- (8,8);
\draw[thick] (0,0) -- (8,8);
\path
(2,6) node(lt) {$F_{i}(x) < F_{j}(x)$}
(7,2) node(eq) {$F_{i}(x) = F_{j}(x)$}
(4.5,1) node(gt) {$F_{i}(x) > F_{j}(x)$};
\draw[->] (gt) .. controls +(-2.5,0) ..  (1.5,1.5);
\draw[->] (gt) .. controls +(2,4) ..  (6,6);
\draw[->] (eq) .. controls +(-2,0) ..  (3.25,3.25);
\draw[->] (eq) .. controls +(0,2) ..  (5,5);
\draw[->] (eq) .. controls +(1,3) ..  (7,7);
\draw[->] (lt) .. controls +(1,-1) .. (4.25,4.25);
\draw[->] (lt) .. controls +(1,2) and +(-1,1) .. (7.5,7.5);
\end{tikzpicture}
\par\end{centering}
\caption{\label{fig:multimaximal-cuts}Illustration for Theorem \ref{prop:multimaximal-cuts}.
The shaded areas and solid lines contain the support of the joint
distribution of $\left(S^{i},S^{j}\right)$ whose split representation
has maximal connections.}
\end{figure}
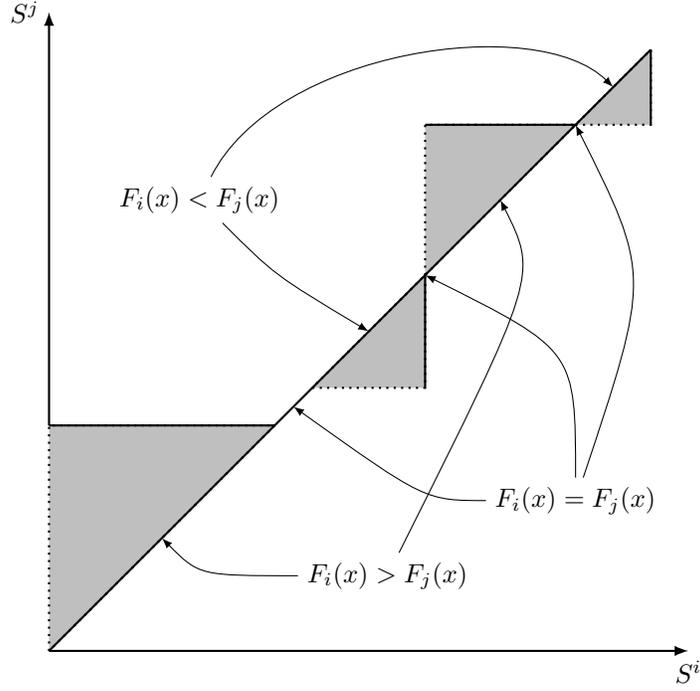

Suppose we have a cut point $x\in\mathbb{R}$ such that $F_{i}(x)>F_{j}(x)$.
For the dichotomized variables $S_{x}^{i}=[S^{i}\leq x]$ and $S_{x}^{j}=[S^{j}\leq x]$
this means $\Pr[S_{x}^{i}=1]>\Pr[S_{x}^{j}=1]$, and by Theorem~\ref{thm:multimax-equi}
(statement \ref{enu:multimax-equi-10}) the pair $\{S_{x}^{i},S_{x}^{j}\}$
is maximal if and only if $\Pr[S_{x}^{i}=0,S_{x}^{j}=1]=0$. This
is equivalent to $\Pr\left[S^{i}>x,\,S^{j}\leq x\right]=0$, i.e.,
the joint distribution of $(S^{i},S^{j})$ vanishes on the set $(x,\infty)\times(-\infty,x]$.

By symmetry, for a cut point $x\in\mathbb{R}$ such that $F_{i}(x)<F_{j}(x)$,
we have $\{S^{i},S^{j}\}$ maximal if and only if the joint distribution
of $(S^{i},S^{j})$ vanishes on the set $(-\infty,x]\times(x,\infty)$.

If $F_{i}(x)=F_{j}(x)$, we have $\Pr[S_{x}^{i}=1]=\Pr[S_{x}^{j}=1]$,
and $\{S^{i},S^{j}\}$ is maximal if and only if $\Pr[S_{x}^{i}=0,S_{x}^{j}=1]=\Pr[S_{x}^{i}=1,S_{x}^{j}=0]=0$,
implying that the joint distribution of $(S^{i},S^{j})$ vanishes
in the set $(-\infty,x]\times(x,\infty)\cup(x,\infty)\times(-\infty,x]$.

Let us denote
\begin{equation}
K=\bigcup_{x\in\mathbb{R}}\underbrace{\begin{cases}
(-\infty,x]\times(x,\infty)\cup(x,\infty)\times(-\infty,x] & \textnormal{if }F_{i}(x)=F_{j}(x),\\
(-\infty,x]\times(x,\infty) & \textnormal{if }F_{i}(x)<F_{j}(x),\\
(x,\infty)\times(-\infty,x] & \textnormal{if }F_{i}(x)>F_{j}(x).
\end{cases}}_{=K_{x}}\label{eq:setK}
\end{equation}
The union does not change if instead of all $x\in\mathbb{R}$, we
take it only over the union of a countable dense subset of $\mathbb{R}$
and the (at most countable set of) boundary points of the sets $\{x\in\mathbb{R}:F_{i}(x)<F_{j}(x)\}$
and $\{x\in\mathbb{R}:F_{i}(x)>F_{j}(x)\}$. Since the distribution
of $(S^{i},S^{j})$ vanishes on each $K_{x}$, it also vanishes for
the countable union of such sets. Thus, we have the following result.
\begin{thm}
\label{prop:multimaximal-cuts}The split representation of a coupling
$\left\{ S^{1},\dots,S^{n}\right\} $ of the single connection $\left\{ R^{1},\dots,R^{n}\right\} $
of a system is multimaximally connected if and only if, for any $i,j\in\left\{ 1,\ldots,n\right\} $,
the joint distribution of $(S^{i},S^{j})$ vanishes on the set $K$
given by (\ref{eq:setK}).
\end{thm}

The region left after removing the set indicated in the theorem is
shown in Figure~\ref{fig:multimaximal-cuts} for a situation when
each of the sets $\left\{ x\in\mathbb{R}:F_{i}(x)=F_{j}(x)\right\} $,
$\left\{ x\in\mathbb{R}:F_{i}(x)<F_{j}(x)\right\} $, and $\left\{ x\in\mathbb{R}:F_{i}(x)>F_{j}(x)\right\} $
is a union of disjoint intervals (including isolated single-point
ones). The statement of Theorem \ref{prop:multimaximal-cuts} holds,
however, in complete generality.

\section{\label{sec:categorical}Split representation for categorical random
variables}

For categorical random variables, Definition \ref{def:allowable}
leads us to use all possible dichotomizations for split representations
of systems. The following notion was introduced in \cite{DzhCerKuj2017}.
\begin{defn}
Given two probability mass functions $p$ and $q$ on the set $\{1,\dots,k\}$
we say that $p$ \emph{nominally dominates} $q$ if and only if $p(i)<q(i)$
for at most one index $i\in\{1,\dots,k\}$. If $A$ and $B$ are random
variables such that the distribution of $A$ nominally dominates the
distribution of $B$ we write $A\succcurlyeq B$.
\end{defn}

The significance of this notion is due to the following result obtained
in \cite{DzhCerKuj2017}. (As we consider single connections of systems,
or single-connection systems, in the remainder of this section, we
continue to drop fixed subscripts indicating contents in their notation,
writing $\left\{ R^{1},\dots,R^{n}\right\} $ instead of $\left\{ R_{q}^{1},\dots,R_{q}^{n}\right\} $.)
\begin{thm}
The split representation of a single connection $\{R^{1},R^{2}\}$
of two categorical random variables with values in $\{1,\dots,$k\}
is noncontextual if and only if $R^{1}\succcurlyeq R^{2}$ or $R^{1}\preccurlyeq R^{2}$.
\end{thm}

Since $k\le3$ implies that $R^{1}\succcurlyeq R^{2}$ or $R^{1}\preccurlyeq R^{2}$
always holds, the split representation of a connection of two categorical
random variables is always noncontextual for $k=3$. For more than
two random variables, this is no longer the case. For instance, we
have the following observation.
\begin{example}
\label{exa:contextual-3-values}The split representation of all possible
dichotomizations of a system consisting of a single connection $\{R^{1},R^{2},R^{3}\}$
with values distributed as 
\[
\begin{array}{cccc}
 & 1 & 2 & 3\\
R^{1} & \nicefrac{1}{2} & \nicefrac{1}{2} & 0\\
R^{2} & 0 & \nicefrac{1}{2} & \nicefrac{1}{2}\\
R^{3} & \nicefrac{1}{2} & 0 & \nicefrac{1}{2}
\end{array}
\]
is contextual. This is the same system that was used as an example
of a set of multi-valued random variables that does not have a multimaximal
coupling in \cite{DzhKuj2.0}--- noncontextuality of the split representation
of all possible dichotomizations of a connection implies the existence
of a multimaximal coupling of the original connection.
\end{example}

Let us consider next a connection $\mathcal{R}=\{R^{1},\dots,R^{n}\}$
with the value set $\{1,\dots,k\}$, $k\ge3$. The split representation
is contextual only if $R^{c}\succcurlyeq R^{c'}$ or $R^{c}\preccurlyeq R^{c'}$
for all pairs $c,c'\in\{1,\dots,n\}$. This is a necessary condition
only. We obtain a rather weak sufficient condition if we impose the
following stringent constraints on the ordering of the probability
distributions. Let us call $\mathcal{R}$ \emph{dominance-aligned}
if, for some permutation $\left\{ k_{1},\ldots,k_{n}\right\} $ of
$\left\{ 1,\ldots,n\right\} $, the ordering
\begin{equation}
\Pr[R^{k_{1}}=i]\le\Pr[R^{k_{2}}=i]\le\dots\le\Pr[R^{k_{n}}=i]
\end{equation}
holds for all but one value of $i\in\left\{ 1,\ldots,n\right\} $.
It is clear that for the exceptional value of $i$ (which, with no
loss of generality, can be taken as $i=1$), the ordering is opposite,
\begin{equation}
\Pr[R^{k_{1}}=i]\geq\Pr[R^{k_{2}}=i]\geq\dots\geq\Pr[R^{k_{n}}=i].
\end{equation}
 Without loss of generality then $\left\{ k_{1},\ldots,k_{n}\right\} $
can be taken to be $\left\{ 1,\ldots,n\right\} $, which we will assume
in the following proposition.
\begin{thm}
If a single connection $\mathcal{R}=\{R^{1},\dots,R^{n}\}$ is dominance-aligned,
the split representation of $\mathcal{R}$ is noncontextual.
\end{thm}

\begin{proof}
Let us consider the split representation consisting of the splits
\[
R_{W}^{c}:=[R^{c}\in W]
\]
for all nonempty $W\subset\{2,\dots,n\}$. All other splits are complements
of these so this is a complete set of splits. Choose a coupling $S$
of $\mathcal{R}$ such that the events 
\[
S^{1}=i,S^{2}=i,\dots,S^{n}=i\tag{*}
\]
form a nested sequence of sets in the domain space of $S$ for each
$i=2,\dots,k$. It follows that the sequence of events
\[
S^{1}\in W,S^{2}\in W,\dots,S^{n}\in W\tag{**}
\]
forms a nested sequence for each nonempty $W\subset\{2,\dots,k\}$,
since these sequences are termwise unions of the nested sequences
({*}). Then, in the split representation of the coupling $S$, each
column $S_{W}$, $W\subset\text{\{2,\ensuremath{\dots,n\}}}$, is
multimaximal, as the sequence of events
\[
S_{W}^{1}=1,S_{W}^{2}=1,\dots,S_{W}^{n}=1
\]
corresponds to the nested sequence ({*}{*}).
\end{proof}
The well-alignedness condition is far from being necessary, as shown
in the example below.
\begin{example}
The split representation of the single-connection system \[%
\begin{tabular}{|c|c|c|c|c|}
\hline 
 & $a$ & $b$ & $c$ & $d$\tabularnewline
\hline 
\hline 
$R^{1}$ & $.7$ & $.1$ & $.1$ & $.1$\tabularnewline
\hline 
$R^{2}$ & .1 & $.5$ & $.2$ & $.2$\tabularnewline
\hline 
$R^{3}$ & $.2$ & $.2$ & $.3$ & $.3$\tabularnewline
\hline 
\end{tabular}\]is noncontextual as it has the coupling\[%
\begin{tabular}{|c|c|c|c|c|c|c|c|c|c|c|}
\hline 
 & $.1$ & $.1$ & $.1$ & $.1$ & $.1$ & $.1$ & $.1$ & $.1$ & $.1$ & $.1$\tabularnewline
\hline 
\hline 
$S^{1}$ & $a$ & $a$ & $a$ & $a$ & $a$ & $a$ & $a$ & $b$ & $c$ & $d$\tabularnewline
\hline 
$S^{2}$ & $b$ & $b$ & $b$ & $b$ & $c$ & $d$ & $a$ & $b$ & $c$ & $d$\tabularnewline
\hline 
$S^{3}$ & $b$ & $a$ & $c$ & $d$ & $c$ & $d$ & $a$ & $b$ & $c$ & $d$\tabularnewline
\hline 
\end{tabular}\] whose split representation can be verified to be multimaximally
connected (see the theorem below for a general condition for this).
However, the exceptional index is $a$ for $R^{2}\succcurlyeq R^{1}$
(and for $R^{3}\succcurlyeq R^{1})$ and $b$ for $R^{3}\succcurlyeq R^{2}$.
\end{example}

\begin{thm}
Let $S=\{S^{1},\dots,S^{n}\}$ be a coupling of the single connection
$\{R^{1},\dots,R^{n}\}$ with the value set $\{1,\dots,k\}$, $k\ge3$,
and let the index $l$ enumerate all value combinations of $S$. Let
$S^{i}(l)$ denote the value of $S^{i}$ in the combination of values
indexed by $l$. With reference to the matrix\[%
\begin{tabular}{|c|c|c|c|c|c|}
\hline 
$ $ & $\cdots$ & $l$ & $\cdots$ & $l'$ & $\cdots$\tabularnewline
\hline 
\hline 
$\vdots$ &  & $\vdots$ &  & $\vdots$ & \tabularnewline
\hline 
$S^{i}$ & $\cdots$ & $S^{i}\left(l\right)=x$ & $\cdots$ & $S^{i}\left(l'\right)=y$ & $\cdots$\tabularnewline
\hline 
$\vdots$ &  & $\vdots$ &  & $\vdots$ & \tabularnewline
\hline 
$S^{i'}$ & $\cdots$ & $S^{i'}\left(l\right)=z$ & $\cdots$ & $S^{i'}\left(l'\right)=w$ & $\cdots$\tabularnewline
\hline 
$\vdots$ &  & $\vdots$ &  & $\vdots$ & \tabularnewline
\hline 
probability & $\cdots$ & $p_{l}$ & $\cdots$ & $p_{l'}$ & $\cdots$\tabularnewline
\hline 
\end{tabular},\] the split representation of $S$ is multimaximally connected
if and only if there are no indices $i,i',l,l'$ such that $p_{l},p_{l'}>0$
and 
\[
x\ne y\ne w\ne z\ne x
\]
(which does not preclude $x=w$ and $z=y$).
\end{thm}

\begin{proof}
If such $i,i',l,l'$ exist, then $\Pr[S_{\{x,w\}}^{i}=1,S_{\{x,w\}}^{i'}=0]\ge p_{l}>0$
and $\Pr[S_{\{x,w\}}^{i}=0,S_{\{x,w\}}^{i'}=1]\ge p_{l'}>0$, which
implies by Theorem~\ref{thm:multimax-equi} (statement \ref{enu:multimax-equi-10})
that $(S_{\{x,w\}}^{1},\dots,S_{\{x,w\}}^{n})$ is not a multimaximal
coupling. Conversely, assume that for some $W\subset\left\{ 1,\ldots,k\right\} $
the coupling $(S_{W}^{1},\dot{\dots,}S_{W}^{n})$ is not multimaximal.
This means, by Theorem~\ref{thm:multimax-equi} (statement \ref{enu:multimax-equi-10}),
that for some $i,i'$ we have $\Pr[S_{W}^{i}=1,S_{W}^{i'}=0]>0$ and
$\Pr[S_{W}^{i}=0,S_{W}^{i'}=1]>0$, which further implies that there
exist indices $l,l'$ with $p_{l},p_{l'}>0$ such that 
\[
\underbrace{S^{i}(l)}_{=:x},\underbrace{S^{i'}(l')}_{=:w}\in W\not\ni\underbrace{S^{i'}(l)}_{=:z},\underbrace{S^{i}(l')}_{=:y}.
\]
This implies $x\ne y\ne w\ne z\ne x$.
\end{proof}
Note that in this proof, if the indices $i,i',l,l'$ satisfying the
stipulated conditions exist, we always can choose $W=\left\{ x,w\right\} $
such that $\left(S_{W}^{i},S_{W}^{i'}\right)$ is not maximal. This
$W$ is a two-element or one-element set. In \cite{DzhCerKuj2017},
a subsystem of a split representation of system consisting of the
one-element or two-element dichotomizations is called a \emph{1-2
subsystem}. We have therefore the following consequence of the theorem.
\begin{cor}
The split representation of a single connection based on all possible
dichotomizations is noncontextual if and only if its 1-2-subsystem
is noncontextual.
\end{cor}

This generalizes the analogous result obtained in \cite{DzhCerKuj2017}
for connections consisting of two random variables.

\section{Conclusion}

To summarize, the CbD notion of (non)contextuality, based on multimaximality
and dichotomizations, satisfies the principles of Analyticity, Noncontextual
Nestedness, Deterministic Redundancy, and Coarse-graining. It properly
specializes to the traditional notion when applied to consistently
connected systems, and contains as a special case the well-developed
theory of cyclic systems of binary random variables. We have formulated
a general principle by which we choose allowable dichotomizations.
It requires that both parts of a dichotomization of the set of possible
values $E_{q}$ be linked subsets of $E_{q}$. For the broad class
of random variables we called ordinary, this uniquely determines the
set $\lambda_{q}$ of all allowable dichotomizations of $E_{q}$,
and we have presented applications of this principle to real-valued
and categorical random variables.

Clearly, we have not provided an exhaustive list of principles or
desiderata for a well-constructed theory of contextuality. We may
need to explicate additional principles to be able to deduce the CbD
theory (or perhaps a generalization thereof) axiomatically, as the
only solution.

\paragraph{Conflict of Interests}

Author declares no conflicts of interests

\paragraph{Data Availability}

Data sharing not applicable to this article as no datasets were generated
or analyzed during the current study.

\end{document}